\documentclass{appolb}
\usepackage{amssymb,latexsym,amsmath,amsfonts,graphicx}
\usepackage{pictex,color,comment}
\newcommand{\bea}{\begin{eqnarray}}
\newcommand{\eea}{\end{eqnarray}}
\newcommand{\beq}{\begin{equation}}
\newcommand{\eeq}{\end{equation}}
\newcommand{\nn}{\nonumber}

\newcommand{\C}[1]{{\mathcal{#1}}}

\newcommand{\lrC}[1]{\left(#1\right)}
\newcommand{\lrS}[1]{\left[#1\right]}

\newcommand{\bigO}{{\C O}}

   \newtheorem{theorem}{Theorem}[section]

    \newtheorem{Definition}[theorem]{Definition}
    
    \newtheorem{Remark}[theorem]{Remark}
    \newenvironment{remark}{\begin{Remark}\rm}{\end{Remark}}
    \newtheorem{Example}[theorem]{Example}
    \newenvironment{example}{\begin{Example}\rm}{\end{Example}}
    \newtheorem{Assumptions}[theorem]{Assumptions}

    \newenvironment{proof}%
    {\rm \trivlist \item[\hskip \labelsep{\bf Proof. }]}%
    {\hspace*{\fill}$\Box$\endtrivlist}
    {\rm \trivlist \item[\hskip \labelsep{\bf Proof}]}%
    {\hspace*{\fill}$\Box$\endtrivlist}

\begin{document}
\title{The Lenard Recursion Relation and a\\Family of Singularly Perturbed Matrix Models
}
\author{Max Atkin
\address{IRMP, Universit\'{e} Catholique de Louvain, Belgium}
}
\maketitle
\begin{abstract}
We review some aspects of recent work concerning double scaling
limits of singularly perturbed hermitian random matrix models and
their connection to Painlev\'{e} equations. We present new results
showing how a Painlev\'{e} III hierarchy recently proposed by the
author can be connected to the Lenard recursion formula used
to construct the Painlev\'{e} I and II hierarchies.
\end{abstract}
  
\section{Introduction}
The semi-classical one hermitian matrix model is defined as the measure
\beq
\frac{1}{Z_n} e^{-n \Tr V(M)}dM
\eeq
on the set $\C{H}^{n\times n}(J)$ of $n\times n$ hermitian matrices whose
spectrum is a subset of intervals $J$. Here the potential $V$ is such that $V'$ is a rational
function. The angular degrees of freedom of $M$ may be integrated out of this
model to give a j.p.d.f for the eigenvalues $x_i$ of $M$,
\beq
\frac{1}{\hat{Z}_n} \Delta(x)^2 \prod^n_{i=1} w(x_i) \chi_J(x_i) dx_i,
\eeq
where $w(x):=e^{-nV(x)}$, $\Delta$ is the Vandemonde determinant and $\chi_J$ is the indicator function on $J$. Such models are a very general class of models for which the method of orthogonal polynomials can be used to give a solution. The method of orthogonal polynomials expresses the eigenvalue $k$-point correlation functions in terms of the correlation kernel,
\beq
K_n(x,y) = h_{n-1}^{-1}\frac{\sqrt{w(x)w(y)}}{x-y} \left(p_n(x)p_{n-1}(y) - p_n(y)p_{n-1}(x) \right),
\eeq
where  $p_j$, $j=0,1,\ldots$ are a family of monic polynomials of degree $j$ characterised by the relations
\beq\label{ortho p}
\int_J p_j(x) p_m(x) w(x) dx = h_j \delta_{j m}.
\eeq
The limiting mean eigenvalue density is given by
\beq
\rho(x) := \lim_{n \rightarrow \infty} \frac{1}{n}K_n(x,x),
\eeq
and describes the macroscopic behaviour of the eigenvalues for large $n$.

Such models have been studied at finite $n$ in \cite{BEH} and their
relation to integrable systems fully described. What is less known is
the types of critical behaviour in such a model. In the case that $V$ is polynomial and $J = \mathbb{R}$, a number of distinct critical points have been identified and studied over the last two decades. These have been classified as;
\begin{itemize}
\item Edge: The spectral density acquires extra zeros at an edge $a$ of its support.
The usual behaviour of $\rho$ near $a$ is $\rho(x) = \bigO(|x-a|^\frac{1}{2})$ however when extra
zeros are present the possible behaviours of $\rho$ are $\rho(x) = \bigO(|x-a|^\frac{4k+1}{2})$ with $k\in \mathbb{N}^0$. The limiting kernel in the neighbourhood of such points is constructed in terms of solutions to the $2k$-th Painlev\'{e} I equation.
\item Interior: The spectral density acquires new zeros at some point $a$ in the interior of its support. The behaviour of $\rho$ near $a$ is $\rho(x) = \bigO(|x-a|^{2k})$ with $k \in \mathbb{N}$. The limiting kernel in the neighbourhood of such points is constructed from solutions to the $k$-th Painlev\'{e} II equation.
\item Exterior: The spectral density acquires a new cut in its support. This transition is known as ``birth of a cut'' \cite{07112609}. Here the limiting kernel is constructed from Hermite polynomials and the local behaviour in the new cut mimics a GUE matrix model.  
\end{itemize}
In the more general case of the semi-classical model no such classification exists, however a number
of special cases have been investigated in the recent literature.

\begin{itemize}
\item The effect of logarithmic singularities in the potential have been investigated in a number of works. In \cite{0305044} the effect of a singularity in the bulk results in a kernel constructed with Bessel functions. In \cite{0508062} the situation of a logarithmic singularity coinciding with a interior critical point was found to lead to kernels containing solutions to the general Painlev\'{e} II equation. Finally logarithmic singularities at the edge of the spectrum results in general Painlev\'{e} XXXIV equations.
\item The addition of a hard edge also results in new behaviour. It has long been known that the kernel near the hard edge can be constructed in terms of Bessel functions. More recent work \cite{0701003} has considered the case of a hard edge meeting a soft edge, with the resulting kernel constructed in terms of Painlev\'{e} XXXIV transcendents. This was further extended in \cite{AA} to a hard edge meeting a edge critical point. It was shown that the associated Painlev\'{e} transcendents satisfy the $k$-th member of the Painlev\'{e} XXXIV hierarchy.
\item Finally, very recently, the behaviour of eigenvalues near poles in the potential have been studied. The case of a simple pole both in the bulk and at the hard edge have been investigated in \cite{CI,XDZ,XDZ2, BMM} and it was shown that the kernel can be constructed using solutions of Painlev\'{e} III. The case of higher order poles at a hard edge was studied in \cite{ACM} by the author and collaborators. The kernel in the neighbourhood of the pole was constructed using solutions of a member of a Painlev\'{e} III hierarchy.
\end{itemize}

In this short note we report on some new aspects of the work in \cite{ACM}. In particular we give a relation between the Painlev\'{e} III hierarchy defined in \cite{ACM} and the Lenard recursion relations that are ubiquitous in the Painlev\'{e} I and II hierarchies.

%%%%%%%%%%%%%%%%%%%%%%%%%%%%%%%%%%%%%%%%%%%%%%%%%%%%%%%%%%%%%%%%%%%
\section{A Painlev\'{e} III Hierarchy}
%%%%%%%%%%%%%%%%%%%%%%%%%%%%%%%%%%%%%%%%%%%%%%%%%%%%%%%%%%%%%%%%%%%
In \cite{ACM} the $k$-th member of the Painlev\'{e} III hierarchy was defined as the system of $k$ ODEs ($p=1,\ldots, k$),
\beq\label{P3def}
\sum_{q=0}^{p} \left(\ell_{k-p+q+1}\ell_{k-q}-(\ell_{k-p+q} \ell_{k-q})'' + 3\ell_{k-p+q}' \ell_{k-q}' - 4u \ell_{k-p+q} \ell_{k-q}\right) = \tau_p,
\eeq
for $k$ unknown functions $\ell_1=\ell_1(s), \ldots, \ell_k=\ell_k(s)$, with $\ell_{k+1}(s) = 0$ and $\ell_0(s) = \frac{s}{2}$. The $\tau_p$'s are constants that act as times. The quantity $u=u(s)$ is defined by,
\beq
u(s) = -\frac{1}{4 \ell_k^2}\left((\ell_k^2)'' - 3 (\ell_k')^2 + \tau_0\right).
\eeq

\begin{example}
For $k=1$ we have the equation
\beq
\ell_1''(s) = \frac{\ell_1'(s){}^2}{\ell_1(s)}-\frac{\ell_1'(s)}{s}-\frac{\ell_1(s){}^2}{s}-\frac{\tau _0}{\ell_1(s)}+\frac{\tau _1}{s},
\eeq
which we identify as a special case of the Painlev\'{e} III equation.
\end{example}
\begin{example}
If $k=2$ we have a system of two ODEs;
\beq
\frac{\tau _1}{2 \ell_1(s) \ell_2(s)}-\frac{\tau _0}{\ell_2(s){}^2}+\frac{\ell_2'(s){}^2}{\ell_2(s){}^2}-\frac{\ell_1'(s) \ell_2'(s)}{\ell_1(s) \ell_2(s)}+\frac{\ell_1''(s)}{\ell_1(s)}-\frac{\ell_2''(s)}{\ell_2(s)}-\frac{\ell_2(s)}{2\ell_1(s)} = 0, 
\eeq
and
\begin{align}
\frac{\ell_1(s){}^2 \ell_2'(s){}^2}{\ell_2(s){}^2}&-\ell_1'(s){}^2+\frac{s \ell_2'(s){}^2}{\ell_2(s)}-\ell_2'(s)-\frac{\tau _0 \ell_1(s){}^2}{\ell_2(s){}^2}-\frac{s \tau _0}{\ell_2(s)}-\tau _2 \nn\\
&=\frac{2\ell_1(s){}^2 \ell_2''(s)}{\ell_2(s)}-2 \ell_1(s) \ell_1''(s)+s \ell_2''(s)+2 \ell_2(s) \ell_1(s).
\end{align}
\end{example}

%%%%%%%%%%%%%%%%%%%%%%%%%%%%%%%%%%%%%%%%%%%%%%%%%%%%%%%%%%%%%%%%%%%
\section{A Riemann-Hilbert Problem for the Painlev\'{e} III Hierarchy}
%%%%%%%%%%%%%%%%%%%%%%%%%%%%%%%%%%%%%%%%%%%%%%%%%%%%%%%%%%%%%%%%%%%

\begin{figure}[t]
\centering 
\includegraphics[scale=0.4]{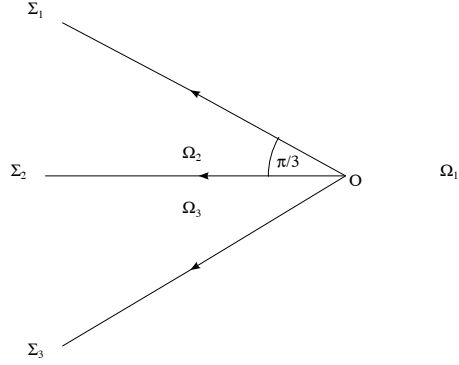}
\caption{The jump contour $\Gamma$ for the model RH problem for $\Psi$. Contours are labelled by $\Sigma$ and sectors by $\Omega$.}
\label{modelcontour}
\end{figure}

In \cite{ACM} it was shown that a solution to the $k$-th Painlev\'{e} III equation may be extracted
from the following RH problem.
\begin{itemize}
\item[(a)] $\Phi:\mathbb C\setminus\Sigma\to\mathbb C^{2\times 2}$ analytic. See Figure \ref{modelcontour}
\item[(b)] $\Phi$ has the jump relations $\Phi_+(z)=\Phi_-(z) j_i$ for $z\in \Sigma_i$
\beq
j_1 = \begin{pmatrix}1&0\\-e^{\pi i\alpha}&1\end{pmatrix}, \quad j_2 = \begin{pmatrix}0&-1\\1&0\end{pmatrix}, \quad j_3 = \begin{pmatrix}1&0\\-e^{-\pi i\alpha}&1\end{pmatrix}.
\eeq
\item[(c)] As $z\to\infty$, $\Phi$ has the asymptotic behaviour
\begin{multline}\label{Phi c}
\Phi(z)=\begin{pmatrix}1&0\\v(s)&1\end{pmatrix} \left(I+\frac{1}{z}\begin{pmatrix}w(s)&v(s)\\h(s)&-w(s)\end{pmatrix} +\bigO(z^{-2})\right)\\
\times e^{\frac{1}{4}i \pi \sigma_3}z^{-\frac{1}{4}\sigma_3}Ne^{sz^{1/2}\sigma_3},
\end{multline}
where $N=\frac{1}{\sqrt{2}}(I+i\sigma_1)$ and $v$, $h$ and $w$ are functions of $s$.
\item[(d)] As $z\to 0$, there exists a matrix $\Phi_0(s)$, independent of $z$, such that $\Phi$ has the asymptotic behaviour
\begin{equation}
\label{Phi0}
\Phi(z)=\Phi_0(s)(I+\bigO(z))e^{\frac{(-1)^{k+1}}{z^{k}}\sigma_3}z^{\frac{\alpha}{2} \sigma_3}H_j,
\end{equation}
for $z \in \Omega_j$, where $H_1, H_2, H_3$ are given by
\beq
H_1=I, \qquad H_2=\begin{pmatrix}1&0\\-e^{\pi i\alpha}&1\end{pmatrix}, \qquad H_3=\begin{pmatrix}1&0\\e^{-\pi i\alpha}&1\end{pmatrix}.
\eeq
\end{itemize}

We then have as a corollary of \cite{ACM} Theorem 1, that
\begin{theorem}
Let $\alpha>-1$, and let $\Phi(z;s)$ be the unique solution of the above model RH problem for $s>0$. Then, the limit
\beq
y(s) = -2 \frac{d}{ds}\left[\lim_{z \rightarrow \infty} s \Phi(z,s) e^{-s z^{1/2}\sigma_3} N^{-1}z^{\frac{1}{4}\sigma_3}e^{-\frac{1}{4}i \pi \sigma_3}\right]_{21}
\eeq
is a solution of the $k$-th member of the Painlev\'{e} III hierarchy
\end{theorem}

\begin{remark}
The proof of this theorem identifies $y$ with $\ell_1$ in the Painlev\'{e} III hierarchy. Furthermore it also demonstrates that the Lax pair for $\Phi$ takes the form,
\begin{align}
\label{ALax}
&A(z,s) = a(z,s) \sigma_3 + b(z,s) \sigma_+ + c(z,s) \sigma_-,\\
\label{BLax}
&B(z;s)=(z-u(s))\sigma_- + \sigma_+.
\end{align}
where $a$, $b$ and $c$ are related by,
\begin{align}
\label{asol}
&a(z,s) = -\frac{1}{2} \partial_s b(z,s),\\
\label{csol}
&c(z,s) = (z-u) b(z,s) - \frac{1}{2} \partial_s^2 b(z,s),\\
\label{aceqn}
&\partial_s c(z,s) = 1 + 2(z-u(s))a(z,s).
\end{align}
Substituting \eqref{asol} and \eqref{csol} into \eqref{aceqn} yields
\beq
z \partial_s b(z,s) = \frac{1}{4} \left(\partial_s^3 b(z,s) + 4 u(s) \partial_s b(z,s) + 2 u'(s) b(z,s) \right) + \frac{1}{2}.
\eeq
We may compute $b(z,s)$ by substituting $b(z,s) = \frac{4}{(4z)^{k+1}}\sum^k_{j=0} \ell_{k-j}(s) (4z)^j$ into the above equation to get,
\beq
\label{lenardtype}
\ell_{j+1}'(s) = \ell_j'''(s) + 4u(s) \ell_j'(s) + 2u'(s)\ell_j(s), 
\eeq
\end{remark}

%%%%%%%%%%%%%%%%%%%%%%%%%%%%%%%%%%%%%%%%%%%%%%%%%%%%%%%%%%%%%%%%%%%
\section{Integration of Lenard-type Recursion Relations}
%%%%%%%%%%%%%%%%%%%%%%%%%%%%%%%%%%%%%%%%%%%%%%%%%%%%%%%%%%%%%%%%%%%

The recursion relation \eqref{lenardtype} is the Lenard recursion relation appearing in the Painlev\'{e} I and II hierarchies. In those cases the initial condition is $\ell_0 = \half$ whereas here we have $\ell_0 = s/2$. Let us consider the general case where $\ell_0$ and $\ell_{k+1}$ are known functions. The fact that $\ell_{k+1}$ is known implies $u(s)$ satisfies a integro-differential equation of order $3k+1$. The following lemma gives $k+1$ constants of motion, i.e. functions of $u(s)$ and its derivatives which are constant in $s$. This allows the equation for $u(s)$ to be reduced to an ODE of order $2k$ and we will see that these constants of motion are precisely the ODEs appearing in \eqref{P3def}.

\begin{theorem}
Let $\ell_j$ be the integro-differential polynomials in $u$ generated by the Lenard recursion relation \eqref{lenardtype} together with an initial condition for $\ell_0$. Furthermore let $\ell_{k+1}$ also have a given form. The integro-differential equation corresponding to $\ell_{k+1}$ has the following constants of motion;
\beq
\tau_{p} = -\ell_{k+1} \ell_{k-p} + \sum^p_{q=0}\lrC{\ell_{k-q}\ell_{k-p+q+1}-\Omega_{k-p+q,k-q}}, \qquad 0 \leq p \leq k
\eeq
if $\ell'_{k+1} = 0$, and
\beq
\sigma_{p} = -\ell_0 \ell_p + \sum^{p-1}_{q=0} \lrC{\Omega_{p-1-q,q} - \ell_{p-1-q}\ell_{q+1}} , \qquad 0 \leq p \leq k
\eeq
if $\ell'_{0} = 0$.
In the above expressions we have introduced,
\beq
\Omega_{n,m}(s) := (\ell_n \ell_m)'' - 3\ell_n' \ell_m' +4u \ell_n \ell_m.
\eeq
\end{theorem}
\begin{proof}
We begin with the following identity,
\beq
\label{masterid}
\ell_m \ell_{n+1}'+ \ell_n \ell_{m+1}' = \Omega_{n,m}'
\eeq
This identity can be established by the following argument,
\begin{align}
\ell_m  ( \ell_{n+1})' &= \ell_m \ell_{n}''' + 4u \ell_m \ell_n' + 2u' \ell_m \ell_n\\
&= \ell_m \ell_{n}''' + 4u \ell_m \ell_n' + 4u' \ell_m \ell_n -\ell_n ( \ell_{m+1}'-4 u \ell_m'-\ell_m''')\\
&= \ell_m \ell_{n}''' + \ell_n \ell_{m}''' + 4(u \ell_m \ell_n)' - \ell_n \ell_{m+1}'\\
&\Rightarrow \ell_m \ell_{n+1}'+\ell_n \ell_{m+1}' = \Omega_{n,m}'
\end{align}
where the first line is the Lenard-type recursion relation multiplied by $\ell_m$, the second line is
obtained by grouping terms and applying the recursion relation for $\ell_m$. The final line is obtained
from the elementary identity,
\beq
\ell_m \ell_{n}''' + \ell_n \ell_{m}''' = (\ell_n \ell_m)'''-3 (\ell_n' \ell_m')'.
\eeq
Now note that integrating \eqref{masterid} by parts and letting $n \mapsto n-1$ gives,
\beq
\label{masterid2}
\ell_m \ell_{n}'= \ell_{m+1} \ell_{n-1}' + \lrS{\Omega_{n-1,m} - \ell_{n-1}\ell_{m+1}}'.
\eeq
Geometrically this equation says that the quantity $\ell_m \ell_{n}'$ only picks up total derivatives if we move along the anti-diagonals of the lattice points labelled by $(m,n)$ for $0\leq m,n \leq k+1$. We can now see how the constants of motion arise. If, due to the boundary conditions on $\ell_0$ and $\ell_{k+1}$, the quantity $\ell_m \ell_{n}'$ is a total derivative on the border of the lattice (i.e. when $m$ or $n$ are $0$ or $k+1$) then we can produce a total derivative equal to zero by using \eqref{masterid2} to move across the lattice from one border to another. Explicitly, moving a distance $r$ along an anti-diagonal gives,
\beq
\ell_m \ell_{n}'= \ell_{m+r} \ell_{n-r}' + \lrS{\sum^{r-1}_{q=0}\Omega_{n-q-1,m+q} - \ell_{n-q-1}\ell_{m+q+1}}'.
\eeq
If $\ell'_{k+1} = 0$ then let $m=k-p$, $n=k+1$ and $r=p+1$. This gives,
\beq
\frac{d}{ds} \lrS{\ell_{k+1} \ell_{k-p} + \sum^p_{q=0}\lrC{\Omega_{k-p+q,k-q} -\ell_{k-q}\ell_{k-p+q+1}}} = 0
\eeq
which integrating gives the first part of the theorem. If $l'_0=0$ then let $m=0$ and $r = n$, this gives,
\beq
\frac{d}{ds} \lrS{\ell_0 \ell_p - \sum^{p-1}_{q=0} \lrC{\Omega_{p-1-q,q} - \ell_{p-1-q}\ell_{q+1}}} = 0,
\eeq
which integrating gives the second part of the theorem.
\end{proof}

\begin{remark}
Setting $\ell_{k+1}=0$ and $\ell=s/2$ in the Lenard recursion relation implies the $\tau_{p}$ are constant and we recover \eqref{P3def}.
\end{remark}

\begin{remark}
Consider the standard Lenard differential polynomials obtained with the boundary condition $\ell_0 = \frac{1}{2}$ and
setting all integration constants to zero. From the above theorem we find that $\sigma_{p}$ are constants, which by the definition of the standard Lenard differential polynomials must be zero. We therefore have, after some rearranging,
\beq
\ell_p = \sum^{p-2}_{q=0} \lrC{\Omega_{p-1-q,q} - \ell_{p-1-q}\ell_{q+1}} + \Omega_{0,p-1}.
\eeq
The right hand side of the above equation only contains $\ell_n$ with $n<p$ while the left hand side only contains $\ell_p$. Hence we can use these equations to recursively determine $\ell_p$, which are exactly the Lenard differential polynomials. Note that this shows each Lenard differential polynomial is indeed a differential polynomial, a fact that is not obvious from their usual definition \eqref{lenardtype}.
\end{remark}

\section*{Acknowledgements} 
MA is supported by the European Research Council under the European Union's Seventh Framework Programme (FP/2007/2013)/ ERC Grant Agreement n.\, 307074, by the Belgian Interuniversity Attraction Pole P07/18, and by F.R.S.-F.N.R.S.

\end{document}